\documentclass[a4paper]{article}
\title{On the lengths of $t$-based confidence intervals}
\author{Yu Zhang$^1$, Xiangzhong Fang$^{1}$\footnote{Corresponding author.  E-mail address:  xzfang@math.pku.edu.cn}\\
1. School of Mathematical Sciences, Peking University, \\ Beijing, 100871, China\\
}
\date{ }

\usepackage{graphicx}
\usepackage[utf8]{inputenc}
\usepackage{diagbox}
\usepackage{color}
\usepackage{xcolor}
\usepackage{amssymb,amsmath,textcomp}
\usepackage{mathrsfs}
\usepackage{float}
\usepackage{cite}
\usepackage{enumerate}
\usepackage{amsthm}
\usepackage{algpseudocode}
\usepackage[bookmarks]{hyperref}
\hypersetup{colorlinks = true, linkcolor = blue}

\usepackage{macros}
\usepackage{algpseudocode}
\usepackage{dsfont}
\usepackage{tikz}
\usepackage{array, booktabs}
\newtheorem{Prop}{Proposition}
\newtheorem{assum}{Assumption}

\begin{document}
\maketitle
\setcounter{tocdepth}{4}
\setcounter{secnumdepth}{4}

\begin{abstract}
Given $n=mk$ $iid$ samples from $N(\theta,\sigma^2)$ with $\theta$ and $\sigma^2$ unknown, we have two ways to construct $t$-based confidence intervals for $\theta$. The traditional method is to treat these $n$ samples as $n$ groups and calculate the intervals. The second, and less frequently used, method is to divide them into $m$ groups with each group containing $k$ elements. For this method, we calculate the mean of each group, and these $k$ mean values can be treated as $iid$ samples from $N(\theta,\sigma^2/k)$. We can use these $k$ values to construct $t$-based confidence intervals. Intuition tells us that, at the same confidence level $1-\alpha$, the first method should be better than the second one. Yet if we define “better” in terms of the expected length of the confidence interval, then the second method is better because the expected length of the  confidence interval obtained from the first method is shorter than the one obtained from the second method. Our work proves this intuition theoretically. We also specify that when the elements in each group are correlated, the first method becomes an invalid method, while the second method can give us correct results. We illustrate this with analytical expressions.
\end{abstract}


\section{Introduction}\label{sec:intro}
In this paper, we consider $t$-based confidence intervals for mean parameter $\theta$ under normal cases. Suppose we have a group of variables $X_1,\cdots,X_n$, which are $iid$ $N(\theta,\sigma^2)$ random variables, where $\theta$ and $\sigma^2$ are both unknown, and the confidence level is $1-\alpha$. A $t$-based confidence interval for $\theta$ can be constructed as follows. We choose the proper pivotal quantity and construct the confidence interval
\begin{align}\label{n-case}
  \left[\overline{X}-t_{n-1,1-\alpha/2}\frac{S_n}{\sqrt{n}},\quad\overline{X}+t_{n-1,1-\alpha/2}\frac{S_n}{\sqrt{n}}\right],
\end{align}
where
\[
\overline{X} = \overline{X}_{*}=\frac{1}{n}\sum_{i=1}^{n}X_i \text{ and } S_n^2=\frac{1}{n-1}\sum_{i=1}^{n}\left(X_i-\overline{X}\right)^2.
\]

Interval (\ref{n-case}) is the traditional way to construct a confidence interval for $\theta$. Now we introduce an alternative.

Suppose we can partition $X_1,\cdots,X_n$ into $m$ equal-sized groups (assuming $n = mk $), and denote the mean of each group by $\overline{X}_1,\cdots,\overline{X}_m$, which are $iid$ $N(\theta,\frac{m}{n}\sigma^2)$. Using these $m$ new variables, the $t$-based confidence interval at level $1-\alpha$ for $\theta$ takes the form
\begin{align}\label{m-case}
\left[\overline{X} - t_{m -1, 1 - \frac{\alpha}{2}} \frac{S_m}{\sqrt m}, \quad \overline{X} + t_{m -1, 1 - \frac{\alpha}{2}} \frac{S_m}{\sqrt m}\right],
\end{align}
where
\[
\overline{X} = \overline{X}_{**} = \frac1m \sum_{i=1}^m \overline X_i = \frac1n \sum_{j=1}^n X_j \text{ and } S^2_m = \frac1{m-1}\sum_{i=1}^m \left(\overline X_i - \overline{X}\right)^2.
\]
(Note that when $m = n$, this reduces to the conventional $t$-based confidence interval, that is, (\ref{m-case}) = (\ref{n-case}).)

We can see that (\ref{n-case}) and (\ref{m-case}) share the same center point, $\overline{X}_{*} = \overline{X}_{**}$, and the only difference is the distance from the center point. Now that there are two ways to obtain the confidence intervals for $\theta$, a question may arise: Are these methods equivalent? Thus far, the literature has not compared the two. Therefore, in this paper we propose a method for comparing them.

To begin this analysis we first need to figure out how to evaluate the confidence intervals. The literature provides several ways to evaluate confidence intervals. In \cite{Pratt1961}, the author proposed that one natural measure of confidence intervals is the expected length. So the expected
length of the interval can be a rule to decide the better result. In \cite{Brown2001}, they proposed alternative intervals for a binomial proportion, and they examined the length of each interval. In \cite{Hall1988}, they proposed another rule coverage probability. However, they didn't deny the validity of the expected length. They thought if we used the coverage as our criterion, it would be possible to construct the "shortest" confidence intervals and these intervals could have much improved coverage accuracy as well as a shorter length. Based on the above we use the expected length of the intervals and the coverage probability as our criteria. We prove theoretically that (\ref{n-case}) is better than (\ref{m-case}) when the samples are $iid$, while if we face non-$iid$ samples, we better use (\ref{m-case}). And we also apply our results to data on HIV patients.

The rest of the paper is organized as follows. Section \ref{sec:main-theorems} introduces our main theorems and related lemmas. In this section, we first present the results when all samples are $iid$. Then we generalize the results to a non-$iid$ case, which is reasonable in real practice. Section \ref{sec:proofs} concludes the detailed proofs of the results in Section \ref{sec:main-theorems}. Section \ref{sec:extensions} presents some extensions related to our results. Section \ref{sec:simulations} contains several simulations, including a real data analysis. Section \ref{sec:discussion} is our discussions.

\section{Main Theorems}\label{sec:main-theorems}

Before proceeding to our main results, we need to define two symbols.
\[
    I_n=2 t_{n -1, 1 - \frac{\alpha}{2}} \frac{S_n}{\sqrt n},
\]
\[
    I_m=2 t_{m -1, 1 - \frac{\alpha}{2}} \frac{S_m}{\sqrt m}
\]
where $I_n$ is the length of (\ref{n-case}) and $I_m$ is the length of (\ref{m-case}).

The main task of our paper is to find out the relationship between $I_n$ and $I_m$. As a first step in this process, we introduce the following trivial proposition related to $I_n^2$ and $I_m^2$.

\begin{Prop}\label{thm:prop}
  \[
    \mathbb{E}I_n^2 \le \mathbb{E}I_m^2.
  \]
\end{Prop}

Proposition \ref{thm:prop} tells us the relationship between $\mathbb{E}I_n^2$ and $\mathbb{E}I_m^2$. However, we
can't conclude that $\mathbb{E}I_n \le \mathbb{E}I_m$ from it directly, as we can with positive real numbers.

The first main result of this paper is Theorem \ref{thm:main}, which demonstrates the relationship between $\mathbb{E}I_n$ and $\mathbb{E}I_m$ when we are dealing with $iid$ samples.

\begin{theorem}\label{thm:main}
  With fixed $0 < \alpha < 1$,
  \begin{align}\label{ieq:EI_n<EI_m}
    \mathbb{E}I_n \le \mathbb{E}I_m.
  \end{align}

  The expected length of $I_m$, i.e., $\mathbb{E}I_m$ is
\[
\begin{aligned}
\mathbb{E} \left[2 t_{m -1, 1 - \frac{\alpha}{2}} \frac{S_m}{\sqrt m} \right] &= 2 t_{m -1, 1 - \frac{\alpha}{2}} \frac{\sqrt{m\sigma^2/n} \mathbb{E} \chi_{m-1}}{\sqrt m \sqrt{m-1}}\\
& = \frac{2\sqrt{2}\sigma}{\sqrt{n}} \cdot \frac{t_{m -1, 1 - \frac{\alpha}{2}} \Gamma(\frac{m}{2})}{\sqrt{m-1}\Gamma(\frac{m-1}{2})}.
\end{aligned}
\]

So, (\ref{ieq:EI_n<EI_m}) is equivalent to the follows. The expression
\[
\frac{t_{m-1, 1-\frac{\alpha}{2}} \Gamma\left(\frac{m}{2}\right)}{\sqrt{m-1}\, \Gamma\left( \frac{m-1}{2}\right)}
\]
as a function of $m$ is decreasing in $m \ge 2$.
\end{theorem}

In order to prove Theorem \ref{thm:main}, we need several lemmas to help us.

\begin{lemma}\label{lm:t_monotone}
Let $1 \le d_1 < d_2$ be two integers. the following expression holds for all $0 < \alpha < 1$,
\begin{equation}\label{eq:t_ratio}
\frac{t_{d_1, 1-\frac{\alpha}{2}}}{t_{d_2, 1-\frac{\alpha}{2}}} > \frac{\sqrt{d_1}\, \Gamma\left( \frac{d_1}{2}\right)\Gamma\left(\frac{d_2+1}{2}\right)}{\sqrt{d_2}\, \Gamma\left( \frac{d_2}{2}\right) \Gamma\left( \frac{d_1+1}{2}\right)}.
\end{equation}
\end{lemma}

\begin{lemma}\label{lm:small_h_pro}
If
\[
1 < \lambda \le \frac{\sqrt{d_1} \Gamma\left(\frac{d_1}{2}\right) \Gamma\left(\frac{d_2+1}{2}\right)}{\sqrt{d_2} \Gamma\left(\frac{d_2}{2}\right) \Gamma\left(\frac{d_1+1}{2}\right)},
\]
then $H_{\lambda}(x) < 0$ for all $x > 0$. Where $H_{\lambda}(x)$ is the difference between $\mathbb{P}(t_{d_1} \le x)$ and $\mathbb{P}(t_{d_2} \le x/\lambda)$ for $x > 0$.
\end{lemma}

If we are able to prove Theorem \ref{thm:main}, does it mean that (\ref{m-case}) is useless? The answer is no. We illustrate this below, focusing specifically on some theorems related to non-$iid$ cases.

In real practice, $X_i$s can't always be $iid$ variables, sometimes there exists some kind of correlations among the part of the variables or the variables can be divided into equal-sized groups and there are correlations within each group, while groups are independent. When dealing with the latter case, we can use (\ref{m-case}) to build the confidence intervals, while (\ref{n-case}) is invalid under this setting.

\begin{assum}\label{as:1}
  Suppose we have $n=mk$ samples. They can be divided into $m$ groups and each group contains $k$ elements. We use $Y_1=(X_1,...,X_k)^T,\cdots,Y_m=(X_{(m-1)k+1},...,X_{mk})^T$ to denote the samples. And we have
$ Y_i \overset{iid}{\sim} N_k(\theta \textbf{1}_k,\Sigma_k) $, where
\[\textbf{1}_k = \left(\begin{array}{c}
                         1 \\
                         \vdots \\
                         1 \\
                         1
                       \end{array}
\right)_{k\times 1} ~and~
\Sigma_k =
\sigma^2\left(
  \begin{array}{cccc}
    1 & \rho_{12} & \ldots & \rho_{1k} \\
    \rho_{21} & 1 & \ddots & \vdots \\
    \vdots & \ddots & \ddots & \rho_{k-1,k} \\
    \rho_{k1} & \cdots & \rho_{k,k-1} & 1 \\
  \end{array}
\right)_{k\times k}
\]
and $\Sigma_k$ is a symmetric positive definite matrix, where all $\rho_{ij} \in (0,1]$.
\end{assum}

We still use $\overline{X}$ to denote the mean of these $n$ samples, that is, $\overline{X}=\frac{X_1+\cdots+X_n}{n}$. Then we have the following generalization.

\begin{theorem}\label{thm:extension}
  Under Assumption \ref{as:1}, the form of the $m$-case interval at level $1 - \alpha$ is also:
  \begin{align*}
    \left[\overline{X} - t_{m -1, 1 - \frac{\alpha}{2}} \frac{S_m}{\sqrt m}, \quad \overline{X} + t_{m -1, 1 - \frac{\alpha}{2}} \frac{S_m}{\sqrt m}\right],
  \end{align*}
  which is the same as (\ref{m-case}). With the similar notation in Section \ref{sec:main-theorems}, we use
  $I_m^g$ to denote the length of the interval, that is, $I_m^g = 2 t_{m -1, 1 - \frac{\alpha}{2}} \frac{S_m}{\sqrt m}$.

  Additionally, we obtain the expected length of the confidence interval under this case:
  \[
  \mathbb{E}I_m^{g} = \frac{2\sqrt{2}}{\sqrt{m(m-1)}}\sqrt{\frac{\Delta}{k^2}}t_{m-1,1-\frac{\alpha}{2}}\frac{\Gamma(\frac{m}{2})}{\Gamma(\frac{m-1}{2})}
  \]
  where $\Delta$ is the sum of the elements in $\Sigma_k$, $\Delta = \sigma^2(k+\sum_{i\neq j}\rho_{ij})$.

  We can extend (\ref{ieq:EI_n<EI_m}) in Theorem \ref{thm:main} to
  \begin{align}\label{ieq:EI_n<EI_m-extend}
    \mathbb{E}(I_n) \le \mathbb{E}(I_m) \le \mathbb{E}(I_m^g)
  \end{align}
\end{theorem}

Now we generalize our results a little bit further.

\begin{assum}\label{as:2}
  Let $X_i = (X_{i1},\cdots,X_{ik})^T
                    \overset{iid}{\sim} N(\theta \textbf{1}_k,\Sigma_k)
  $ $(i=1,\cdots,m)$, $n=mk$,
  where $\textbf{1}_k$ and $\Sigma_k$ have the same form in Assumption \ref{as:1}.
\end{assum}

We define two different sample variances here. One is
\[
S_n^2=\frac{1}{n-1}\sum\limits_{i=1}^{m}\sum\limits_{j=1}^{k}(X_{ij}-\overline{X})^2
\]
where $\overline{X}=\frac{1}{n}\sum\limits_{i=1}^{m}\sum\limits_{j=1}^{k}X_{ij}$. The other is
\[
S_m^2=\frac{1}{m-1}\sum\limits_{i=1}^{m}\sum\limits_{j=1}^{k}(\overline{X_i}-\overline{X})^2
\]
where $\overline{X_i}=\frac{1}{k}\sum\limits_{j=1}^{k}X_{ij}$.

The next theorem tells us under Assumption \ref{as:2}, if we calculate $n$-case interval, then we can't obtain the desired coverage probability.
\begin{theorem}\label{thm:le 1- alpha}
 Under Assumption \ref{as:2},

  \[\lim\limits_{n\rightarrow\infty}\mathbb{P}(\overline{X}-t_{n-1,1-\frac{\alpha}{2}}\frac{S_n}{\sqrt{n}}
  \le \theta \le
  \overline{X}+t_{n-1,1-\frac{\alpha}{2}}\frac{S_n}{\sqrt{n}}) < 1-\alpha\]
\end{theorem}

We use Lemma \ref{neq:lemma-1} to help us finish the proof of Theorem \ref{thm:le 1- alpha}.
\begin{lemma}\label{neq:lemma-1}
  \begin{equation}
    \frac{1}{k}\textbf{1}_k^T\Sigma_k\textbf{1}_k > S_n^2
  \end{equation}
\end{lemma}

Theorem \ref{thm:le 1- alpha} tells us that the probability is less than $1-\alpha$. The next theorem shows a more accurate result.

\begin{theorem}\label{thm:exact}
Under Assumption \ref{as:2},
  \begin{equation}\label{eq:exact}
  \lim\limits_{n\rightarrow\infty}\mathbb{P}\left(\left|\frac{\sqrt{n}(\overline{X}-\theta)}{S_n}\right|\le t_{n-1,1-\frac{\alpha}{2}}\right)=2\Phi\left(z_{1-\frac{\alpha}{2}}\left|\frac{S_n}{S_m\sqrt{k}}\right|\right)-1
  \end{equation}
  where $S_n$ and $S_m$ are the same as defined in (\ref{n-case}) and (\ref{m-case}).
\end{theorem}

\section{Proofs}\label{sec:proofs}

In this section, we provide detailed proofs of each theorem introduced in Section \ref{sec:main-theorems}. We also prove related lemmas for each theorem.

First, we provide the proof of Proposition \ref{thm:prop}.

\begin{proof}[Proof of Proposition \ref{thm:prop}]

    Note that $S_n^2$ is the sample variance and it is an unbiased estimator of the population variance. So we have $\mathbb{E}(S_n^2)=\sigma^2$ and the same result with $S_m^2$, $\mathbb{E}(S_m^2)=\frac{m}{n}\sigma^2$.
    \[
    \mathbb{E}(I_n^2)=t^2_{n-1,1-\alpha/2}\frac{4\sigma^2}{n}
    \]
    \[
    \mathbb{E}(I_m^2)=t^2_{m-1,1-\alpha/2}\frac{4\sigma^2}{n}
    \]
    We only need to compare $t^2_{n-1,1-\alpha/2}$ and $t^2_{m-1,1-\alpha/2}$. Since $n > m$, we have $t_{m-1,1-\alpha/2} > t_{n-1,1-\alpha/2} >0$. Thus, we have the desired result.
\end{proof}

Now we proceed to the proof of Theorem \ref{thm:main} and  present the proofs of Lemma \ref{lm:t_monotone} and Lemma \ref{lm:small_h_pro}.

Theorem \ref{thm:main} amounts to showing that
\[
\frac{t_{m-1, 1-\frac{\alpha}{2}} \Gamma\left(\frac{m}{2}\right)}{\sqrt{m-1}\, \Gamma\left( \frac{m-1}{2}\right)} > \frac{t_{m, 1-\frac{\alpha}{2}} \Gamma\left(\frac{m+1}{2}\right)}{\sqrt{m}\, \Gamma\left( \frac{m}{2}\right)}
\]
for every $m \ge 2$. More generally, we aim to show that
\[
\frac{t_{d_1, 1-\frac{\alpha}{2}} \Gamma\left(\frac{d_1+1}{2}\right)}{\sqrt{d_1}\, \Gamma\left( \frac{d_1}{2}\right)} > \frac{t_{d_2, 1-\frac{\alpha}{2}} \Gamma\left(\frac{d_2+1}{2}\right)}{\sqrt{d_2}\, \Gamma\left( \frac{d_2}{2}\right)}
\]
for all $d_2 > d_1 \ge 1$. This inequality can be rewritten as (\ref{eq:t_ratio}).
To establish \eqref{eq:t_ratio} is the subject of Lemma \ref{lm:t_monotone}.

\begin{proof}[Proof of Lemma \ref{lm:t_monotone}]
Suppose on the contrary that
\begin{equation}\label{eq:assume_wro}
\frac{t_{d_1, 1-\frac{\alpha_0}{2}}}{t_{d_2, 1-\frac{\alpha_0}{2}}} \le \frac{\sqrt{d_1} \Gamma\left(\frac{d_1}{2}\right) \Gamma\left(\frac{d_2+1}{2}\right)}{\sqrt{d_2} \Gamma\left(\frac{d_2}{2}\right) \Gamma\left(\frac{d_1+1}{2}\right)}
\end{equation}
for some $\alpha_0 \in (0,1)$. Denote by $\lambda = \frac{t_{d_1, 1-\frac{\alpha_0}{2}}}{t_{d_2, 1-\frac{\alpha_0}{2}}}$. Note that $\lambda$ must be larger than 1 since $d_2 > d_1$.
Write $H_{\lambda}(x)$ for the difference between $\mathbb{P}(t_{d_1} \le x)$ and $\mathbb{P}(t_{d_2} \le x/\lambda)$ for $x > 0$. In particular, this function satisfies
\begin{equation}\label{eq:H_zero}
\begin{aligned}
H_{\lambda}(t_{d_1, 1 - \frac{\alpha_0}{2}}) &= \mathbb{P}(t_{d_1} \le t_{d_1, 1 - \frac{\alpha_0}{2}}) - \mathbb{P}(t_{d_2} \le t_{d_1, 1 - \frac{\alpha_0}{2}}/\lambda)\\
& = \mathbb{P}(t_{d_1} \le t_{d_1, 1 - \frac{\alpha_0}{2}}) - \mathbb{P}(t_{d_2} \le t_{d_2, 1 - \frac{\alpha_0}{2}})\\
& = 1 - \frac{\alpha_0}{2} - \left(1 - \frac{\alpha_0}{2} \right)\\
& = 0.
\end{aligned}
\end{equation}
In general, denoting the density of $t_d$ by $p_d$, $H_{\lambda}(x)$ takes the following form
\[
\begin{aligned}
H_{\lambda}(x) &= \mathbb{P}(t_{d_1} \le x) - \mathbb{P}(t_{d_2} \le x/\lambda)\\
& = \int^{x}_{-\infty} p_{d_1}(u) \d u  - \int^{x/\lambda}_{-\infty} p_{d_2}(u) \d u\\
& = \int^{x}_{0} p_{d_1}(u) \d u  - \int^{x/\lambda}_{0} p_{d_2}(u) \d u\\
& = \int^{x}_{0} p_{d_1}(u) - p_{d_2}(u/\lambda)/\lambda \d u.
\end{aligned}
\]
Let $h_{\lambda}(u)$ be the integrand $p_{d_1}(u) - p_{d_2}(u/\lambda)/\lambda$. If one can show that
\begin{equation}\label{eq:H_last_step}
H_{\lambda}(x) = \int_0^{x} h_{\lambda}(u) \d u < 0
\end{equation}
for all $x > 0$ given
\[
\lambda \le \frac{\sqrt{d_1} \Gamma\left(\frac{d_1}{2}\right) \Gamma\left(\frac{d_2+1}{2}\right)}{\sqrt{d_2} \Gamma\left(\frac{d_2}{2}\right) \Gamma\left(\frac{d_1+1}{2}\right)}
\]
that follows from the assumption \eqref{eq:assume_wro}, we get a contradiction to \eqref{eq:H_zero}. Consequently, \eqref{eq:assume_wro} cannot be satisfied.

Below, Lemma \ref{lm:small_h_pro} affirms \eqref{eq:H_last_step} as the last step to prove the present lemma, thus, concluding the proof of Theorem \ref{thm:main}.

\end{proof}

\begin{proof}[Proof of Lemma \ref{lm:small_h_pro}]
Note that $H_{\lambda}(x)$ results from integrating $h_{\lambda}(x)$. The sign of $h_{\lambda}(x)$ depends on whether the ratio
\[
\begin{aligned}
\frac{\frac{\Gamma\left(\frac{d_1+1}{2} \right)}{\sqrt{\pi d_1}\Gamma\left(\frac{d_1}{2} \right)} \left(1 + \frac{x^2}{d_1}\right)^{-\frac{d_1+1}{2}}}{\frac{\Gamma\left(\frac{d_2+1}{2} \right)}{\lambda\sqrt{\pi d_2}\Gamma\left(\frac{d_2}{2} \right)} \left(1 + \frac{x^2}{\lambda^2 d_2}\right)^{-\frac{d_2+1}{2}}} & = \frac{\lambda \sqrt{d_2} \Gamma\left(\frac{d_2}{2}\right) \Gamma\left(\frac{d_1+1}{2}\right)}{\sqrt{d_1} \Gamma\left(\frac{d_1}{2}\right) \Gamma\left(\frac{d_2+1}{2}\right)} \cdot \frac{\left(1 + \frac{x^2}{\lambda^2 d_2}\right)^{\frac{d_2+1}{2}}}{\left(1 + \frac{x^2}{d_1}\right)^{\frac{d_1+1}{2}}}\\
& \equiv C_{\lambda} \frac{\left(1 + \frac{x^2}{\lambda^2 d_2}\right)^{\frac{d_2+1}{2}}}{\left(1 + \frac{x^2}{d_1}\right)^{\frac{d_1+1}{2}}}\\
\end{aligned}
\]
exceeds 1 or not. Above, we use the fact that the density of the $t$-based distribution with $d$ degrees of freedom reads
\[
p_d(x) = \frac{\Gamma\left(\frac{d+1}{2} \right)}{\sqrt{\pi d}\Gamma\left(\frac{d}{2} \right)} \left(1 + \frac{x^2}{d}\right)^{-\frac{d+1}{2}}
\]
for $x > 0$. Denote by this ratio $r_{\lambda}(x)$. It is clear that $r_\lambda(0) = C_{\lambda} \le 1$, implying that $h_{\lambda}(0) \le 0$. In addition, as $x \rightarrow \infty$, the ratio $r_{\lambda}(x) \rightarrow \infty$ as well, and this reveals that $h_{\lambda}(x) > 0$ for sufficiently large $x$.

To get a closer look, note that
\[
\begin{aligned}
\frac{\d \log r_{\lambda}(x)}{\d x} & = \frac{d_2 + 1}{2} \frac{2x}{\lambda^2 d_2 + x^2} - \frac{d_1 + 1}{2} \frac{2x}{d_1 + x^2}\\
& = \frac{(d_2-d_1)x^3 - \left[(\lambda^2 - 1)d_1d_2 + \lambda^2 d_2 - d_1 \right] x}{(\lambda^2 d_2 + x^2)(d_1 + x^2)}.
\end{aligned}
\]
The fact that $\lambda > 1$ and $d_2 > d_1$ ensures that $(\lambda^2 - 1)d_1d_2 + \lambda^2 d_2 - d_1 > 0$. Hence, we get
\[
\begin{aligned}
\frac{\d \log r_{\lambda}(x)}{\d x}  < 0, & \text{ if } 0 < x < \sqrt{\frac{(\lambda^2 - 1)d_1d_2 + \lambda^2 d_2 - d_1}{d_2 - d_1}}\\
\frac{\d \log r_{\lambda}(x)}{\d x}  > 0, & \text{ if } x > \sqrt{\frac{(\lambda^2 - 1)d_1d_2 + \lambda^2 d_2 - d_1}{d_2 - d_1}}.
\end{aligned}
\]
Thus, $r_{\lambda}(x)$, starting from $r_{\lambda}(0) \le 1$, stays below 1 for $0 < x < x_0$ and then stays above 1 for $x_0 < x < \infty$, where $x_0 > 0$ is some number determined by $d_1, d_2$ and $\lambda$. Put differently, the above discussion demonstrates that
\begin{equation}\label{eq:H_sign}
\begin{aligned}
h_{\lambda}(x) < 0, &\text{ for } 0 < x < x_0\\
h_{\lambda}(x) > 0, &\text{ for } x_0< x < \infty.
\end{aligned}
\end{equation}

Having established \eqref{eq:H_sign}, it is a stone's throw away to prove the lemma. If $x < x_0$, then \eqref{eq:H_sign} readily gives
\[
H_{\lambda}(x) = \int_0^{x} h_{\lambda}(u) \d u < 0.
\]
In the case where $x > x_0$, \eqref{eq:H_sign} together with the fact that $H_{\lambda}(\infty) = 0$ gives
\[
H_{\lambda}(x) = \int_0^{x} h_{\lambda}(u) \d u = \int_0^{\infty} h_{\lambda}(u) \d u - \int_{x}^{\infty} h_{\lambda}(u) \d u = - \int_{x}^{\infty} h_{\lambda}(u) \d u < 0,
\]
as desired.

\end{proof}

Now we begin to prove the generalized results.

\begin{proof}[Proof of Theorem \ref{thm:extension}]
  First, we need to find out the distribution of $\textbf{1}_k'Y_i$.

  Since $Y_i \sim N_k(\theta \textbf{1}_k,\Sigma_k)$, for $\textbf{1}_k'Y_i$, we have
  \[
  \textbf{1}_k'Y_i \sim N(\textbf{1}_k'\theta\textbf{1}_k,\textbf{1}_k'\Sigma_k\textbf{1}_k)\sim N(k\theta,\Delta)
  \]

Note that $\textbf{1}_k'Y_i$ is the sum of the elements in $Y_i$, namely, $\textbf{1}_k'Y_i=X_{(i-1)k+1}+\cdots+X_{ik}$.

Secondly, we can proceed to our result. Let $Z_i=\frac{1}{k}\textbf{1}_k'Y_i$. $Z_i \overset{iid}{\sim} N(\theta,\frac{\Delta}{k^2})$. We also use $S^2_m = \frac1{m-1}\sum_{i=1}^m \left(Z_i - \overline{Z}\right)^2$ to denote the sample variance.

It's easy to get that \[\mathbb{E}S_m=\sqrt{2}\sqrt{\frac{1}{m-1}}\sqrt{\frac{\Delta}{k^2}}\frac{\Gamma(\frac{m}{2})}{\Gamma(\frac{m-1}{2})}\]
where $S_m$ is actually a variant of a $\chi$ random variable.

So the expected length of the $m$-case interval is
\[
\mathbb{E}I_m^{g} = 2t_{m-1,1-\frac{\alpha}{2}}\frac{\mathbb{E}S_m}{\sqrt{m}}= \frac{2\sqrt{2}}{\sqrt{m(m-1)}}\frac{\sqrt{\Delta}}{k}t_{m-1,1-\frac{\alpha}{2}}\frac{\Gamma(\frac{m}{2})}{\Gamma(\frac{m-1}{2})}
  \]
\end{proof}

Before proceeding to the proofs of Theorem \ref{thm:le 1- alpha} and Theorem \ref{thm:exact}, we provide the proof of Lemma \ref{neq:lemma-1} first.

\begin{proof}[Proof of Lemma \ref{neq:lemma-1}]

  \begin{align}\label{neq:lemma-1-proof}
  \begin{split}
    \mathbb{E}S_n^2 & = \mathbb{E}\left\{\frac{1}{n-1}\sum\limits_{i=1}^{m}\left[\sum\limits_{j=1}^{k}(X_{ij}-\overline{X_i})^2+k(\overline{X_i}-\overline{X})^2\right]\right\} \\
     & = \frac{m}{n-1}\mathbb{E}[X_i^T(I-\frac{1}{k}\textbf{1}_k\textbf{1}_k^T)X_i]+\frac{(m-1)k}{n-1}\mathbb{E}S_m^2\\
     & = \frac{m}{n-1}\left\{tr\left[(I-\frac{1}{k}\textbf{1}_k\textbf{1}_k^T)\Sigma_k\right]+(\theta\textbf{1}_k)^T(I-\frac{1}{k}\textbf{1}_k\textbf{1}_k^T)(\theta\textbf{1}_k)\right\}+\frac{(m-1)k}{n-1}\frac{1}{k^2}\textbf{1}_k^T\Sigma_k\textbf{1}_k\\
     & = \frac{m}{n-1}tr\left[(I-\frac{1}{k}\textbf{1}_k\textbf{1}_k^T)\Sigma_k\right]+\frac{(m-1)}{n-1}\frac{1}{k}\textbf{1}_k^T\Sigma_k\textbf{1}_k\\
     & = \frac{m(k-1)}{n-1}\frac{k^2-\Delta/\sigma^2}{k}\sigma^2+\frac{m-1}{n-1}\frac{\Delta/\sigma^2}{k}\sigma^2\\
     & = \frac{mk^2-\Delta/\sigma^2}{(n-1)k}\sigma^2\\
     & = \sigma^2 + o(1)
    \end{split}
  \end{align}

  Thus, $\mathbb{E}S_n^2 < \frac{1}{k}\textbf{1}_k^T\Sigma_k\textbf{1}_k$. With the help of the equation
  \begin{equation}\label{eq:fact}
    S_n^2 = \mathbb{E}S_n^2 + o_p(1),
  \end{equation}
  we can have the desired lemma.
\end{proof}

\begin{proof}[Proof of Theorem \ref{thm:le 1- alpha}]
From the Assumption \ref{as:2}, we have
  \begin{equation}\label{eq:thm1-1}
  \overline{X_i} \sim N(\theta,\frac{1}{k^2}\textbf{1}_k^T\Sigma_k\textbf{1}_k)
  \end{equation}

and \begin{equation}\label{eq:thm1-2}
        \begin{split}
           \sqrt{n}(\overline{X}-\theta) & = \frac{1}{\sqrt{n}}\sum\limits_{i=1}^{m}\sum\limits_{j=1}^{k}(X_{ij}-\theta) \\
       & =\frac{\sqrt{m}}{\sqrt{n}}\frac{1}{\sqrt{m}}\sum\limits_{i=1}^{m}k(\overline{X_i}-\theta)\\
       & = \sqrt{\frac{1}{k}}\frac{1}{\sqrt{m}}\sum\limits_{i=1}^{m}k(\overline{X_i}-\theta)\\
       & = \sqrt{\frac{1}{n}}\sum\limits_{i=1}^{m}k(\overline{X_i}-\theta)
        \end{split}
    \end{equation}
Combine (\ref{eq:thm1-1}) and (\ref{eq:thm1-2}) together, we have
\begin{equation}\label{eq:thm1-3}
  \sqrt{n}(\overline{X}-\theta) \sim N(0,\frac{1}{k}\textbf{1}_k^T\Sigma_k\textbf{1}_k) \Leftrightarrow
  \frac{\sqrt{n}(\overline{X}-\theta)}{\sqrt{\frac{1}{k}\textbf{1}_k^T\Sigma_k\textbf{1}_k}} \sim N(0,1)
\end{equation}

Before having our final result, we still need one expression.
\begin{equation}\label{eq:thm1-4}
  \left|\frac{\sqrt{n}(\overline{X}-\theta)}{S_n}\right|=
  \left|\frac{\sqrt{n}(\overline{X}-\theta)}{\sqrt{\frac{1}{k}\textbf{1}_k^T\Sigma_k\textbf{1}_k}}\right|
  \cdot
  \left|\sqrt{\frac{\frac{1}{k}\textbf{1}_k^T\Sigma_k\textbf{1}_k}{S_n^2}}\right|
\end{equation}

With (\ref{neq:lemma-1}) and (\ref{eq:thm1-4}), now we have, with $n \rightarrow \infty$
\begin{equation}\label{eq:thm1-5}
  \begin{split}
     \lim\limits_{n\rightarrow\infty}\mathbb{P}\left(\left|\frac{\sqrt{n}(\overline{X}-\theta)}{S_n}\right|\le t_{n-1,1-\frac{\alpha}{2}}\right) & < \lim\limits_{n\rightarrow\infty}\mathbb{P}\left(\left|\frac{\sqrt{n}(\overline{X}-\theta)}{\sqrt{\frac{1}{k}\textbf{1}_k^T\Sigma_k\textbf{1}_k}}\right|\le t_{n-1,1-\frac{\alpha}{2}}\right) \\
   & = \lim\limits_{n\rightarrow\infty}\mathbb{P}\left(\left|\frac{\sqrt{n}(\overline{X}-\theta)}{\sqrt{\frac{1}{k}\textbf{1}_k^T\Sigma_k\textbf{1}_k}}\right|\le z_{1-\frac{\alpha}{2}}\right) \\
   & = 1-\alpha
  \end{split}
\end{equation}
where $z_{1-\frac{\alpha}{2}}$ is the $1-\frac{\alpha}{2}$ percentile of $N(0,1)$.
\end{proof}

\begin{proof}[Proof of Theorem \ref{thm:exact}]
  From (\ref{neq:lemma-1-proof}) and (\ref{eq:fact}), we can conclude that
  \begin{align}\label{thm:exact-proof}
    S_n^2=\sigma^2 + o_p(1)
  \end{align}
  With the help of (\ref{thm:exact-proof}), (\ref{eq:thm1-3}) and Slutsky's Theorem, we have
  \[
  \frac{\sqrt{n}(\overline{X}-\theta)}{\sqrt{\Delta/k}}\frac{\sigma}{S_n}\overset{L}{\rightarrow}N(0,1).
  \]
  Then the LHS of (\ref{eq:exact})
  \begin{equation}
  \begin{split}
    & = \lim\limits_{n\rightarrow\infty}\mathbb{P}\left(\left|\frac{\sqrt{n}(\overline{X}-\theta)}{S_n}\right|\le z_{1-\frac{\alpha}{2}}\right) \\
    & = \lim\limits_{n\rightarrow\infty}\mathbb{P}\left(\left|\frac{\sqrt{n}(\overline{X}-\theta)}{\sqrt{\Delta/k}}\right|
    \left|\frac{\sqrt{\Delta/k}}{S_n}\right|
    \le z_{1-\frac{\alpha}{2}}\right)\\
    & = \lim\limits_{n\rightarrow\infty}\mathbb{P}\left(\left|\frac{\sqrt{n}(\overline{X}-\theta)}{\sqrt{\Delta/k}}\right|
    \le z_{1-\frac{\alpha}{2}}\left|\frac{S_n}{\sqrt{\Delta/k}}\right|\right)
  \end{split}
  \end{equation}
where $\Delta$ is unknown under our assumption. However, from (\ref{eq:thm1-3}), we know that
\[
\mathbb{E}S_m^2=\frac{\Delta}{k^2}.
\]
So we can obtain the estimate of $\Delta$,
\[
\widehat{\Delta}=k^2S_m^2.
\]
The LHS of (\ref{eq:exact}) now becomes
\[
2\Phi\left(z_{1-\frac{\alpha}{2}}\left|\frac{S_n}{S_m\sqrt{k}}\right|\right)-1
\]
\end{proof}

\section{Extensions}\label{sec:extensions}

From the theorems in Section \ref{sec:main-theorems}, we can learn that with a fixed level $1-\alpha$, (\ref{n-case}) provides shorter interval than (\ref{m-case}) when the samples are $iid$. We have proposed a special case (Assumption \ref{as:1}) in which the coverage probability of (\ref{n-case}) is less than $1-\alpha$. Luckily, (\ref{m-case}) can provide us the right confidence interval. In this section, we discuss a little bit more about (\ref{m-case}). We propose some situations, in which (\ref{m-case}) can be properly utilized, if we are restricted to the time cost or the equipment.

\subsubsection*{Situation 1}

Suppose that our data are collected from $m$ locations. Gathering the data from all locations is time-consuming or impossible. In this case, we can only use (\ref{m-case}) to calculate the confidence interval. We first calculate the mean of the data in each location and then transfer these mean values to a center location. Finally, we can use (\ref{m-case}) to get the confidence interval with these mean values.

\subsubsection*{Situation 2}

In this situation, the data are stored in one machine. However the scale of the data is extremely large. Due to the restriction of the hardware, calculating the confidence interval with (\ref{n-case}) is impossible. We may split the data into $m$ equal-sized groups, calculate the mean and perform (\ref{m-case}) simultaneously.

\section{Simulations}\label{sec:simulations}

In this section, we conduct several simulations as auxiliary validations for our theorems.

\subsection{Simulation 1}\label{Sec:Simu1}

Suppose the sample size is 420 and all the samples are $iid$ from $N(0,1)$. The confidence level $1-\alpha$ is 0.95.

420 has 23 factors except 1. We let these 23 values be the number of groups, that is, $m$. We compute the length of the confidence intervals with (\ref{m-case}) and we obtain 23 values.

We repeat the above procedure 100 times and obtain 100 values for each $m$. Then we take the average of each set of 100 values and get Table \ref{table:E-length-1}. Note that when $m=420$, it's the $n$-case interval.

\begin{table}[htb]
\centering
\begin{tabular}{|c|c|c|c|c|c|c|c|c|c|c|c|c|}
\hline
m & 420  & 210  & 140  & 105  & 84   & 70 \\ \hline
$\mathbb{E}I_m$ & 0.19 & 0.19 & 0.19 & 0.19 & 0.19 & 0.19  \\ \hline
m & 60   & 42   & 35   & 30   & 28    & 21 \\ \hline
$\mathbb{E}I_m$ & 0.19 & 0.19 & 0.19 & 0.19 & 0.19  & 0.20\\ \hline
m & 20   & 15   & 14   & 12   & 10   & 7  \\ \hline
$\mathbb{E}I_m$ & 0.20 & 0.20 & 0.20 & 0.20 & 0.21 & 0.22 \\ \hline
m & 6    & 5    & 4    & 3    & 2 & \\ \hline
$\mathbb{E}I_m$ & 0.23 & 0.24 & 0.28 & 0.34 & 0.84 & \\\hline
\end{tabular}
\caption{Average lengths of confidence intervals for each $m$}\label{table:E-length-1}
\end{table}

Figure \ref{fig:Simu1} shows the results intuitively. The horizontal axis is the values of $m$, we have rescaled the axis to make the figure easy to understand. The vertical axis is the expected lengths of the confidence intervals. From the figure we can see the decreasing pattern in the lengths with increasing $m$.

\subsection{Simulation 2}\label{Sec:Simu2}

We use the same settings as those in \nameref{Sec:Simu1}, except that samples are from $N(0,100)$. Figure \ref{fig:Simu2} shows the results intuitively and more detailed results are presented in Table \ref{table:E-length-2}.
\begin{table}[htb]
\centering
\begin{tabular}{|c|c|c|c|c|c|c|}
\hline
m & 420  & 210  & 140  & 105  & 84   & 70 \\ \hline
$\mathbb{E}I_m$ & 19.15 & 19.23 & 19.25 & 19.34 & 19.45 & 19.45  \\ \hline
m & 60   & 42   & 35   & 30   & 28    & 21 \\ \hline
$\mathbb{E}I_m$ &19.55 & 19.57 & 19.92 & 20.17 & 20.20  & 20.33\\ \hline
m & 20   & 15   & 14   & 12   & 10   & 7  \\ \hline
$\mathbb{E}I_m$ & 20.59 & 20.78 & 21.03 & 21.10 & 21.73 & 23.34 \\ \hline
m & 6    & 5    & 4    & 3    & 2 & \\ \hline
$\mathbb{E}I_m$ &  24.36 & 26.26 & 29.31 & 38.80 & 100.29 & \\\hline
\end{tabular}
\caption{Average length of confidence intervals for each $m$}\label{table:E-length-2}
\end{table}

Please note that Figure \ref{fig:Simu1} and Figure \ref{fig:Simu2} look similar, while the ranges of their vertical axes are different.

\begin{figure}[htbp]
\centering
\begin{minipage}[t]{0.48\textwidth}
\centering
\includegraphics[width=\textwidth]{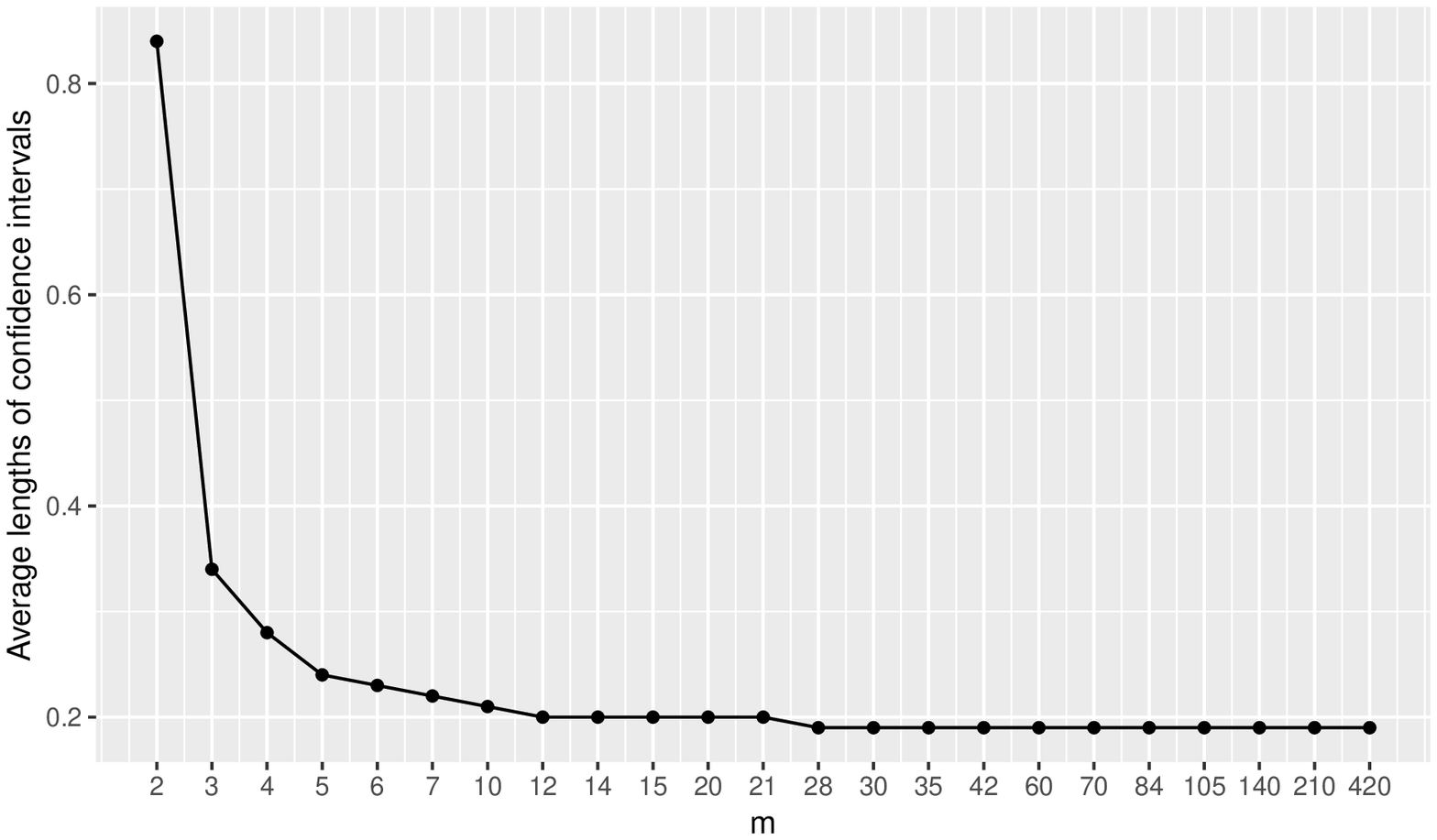}
  \caption{The Expected lengths of Confidence Intervals in \nameref{Sec:Simu1}}\label{fig:Simu1}
\end{minipage}
\begin{minipage}[t]{0.48\textwidth}
\centering
  \includegraphics[width=\textwidth]{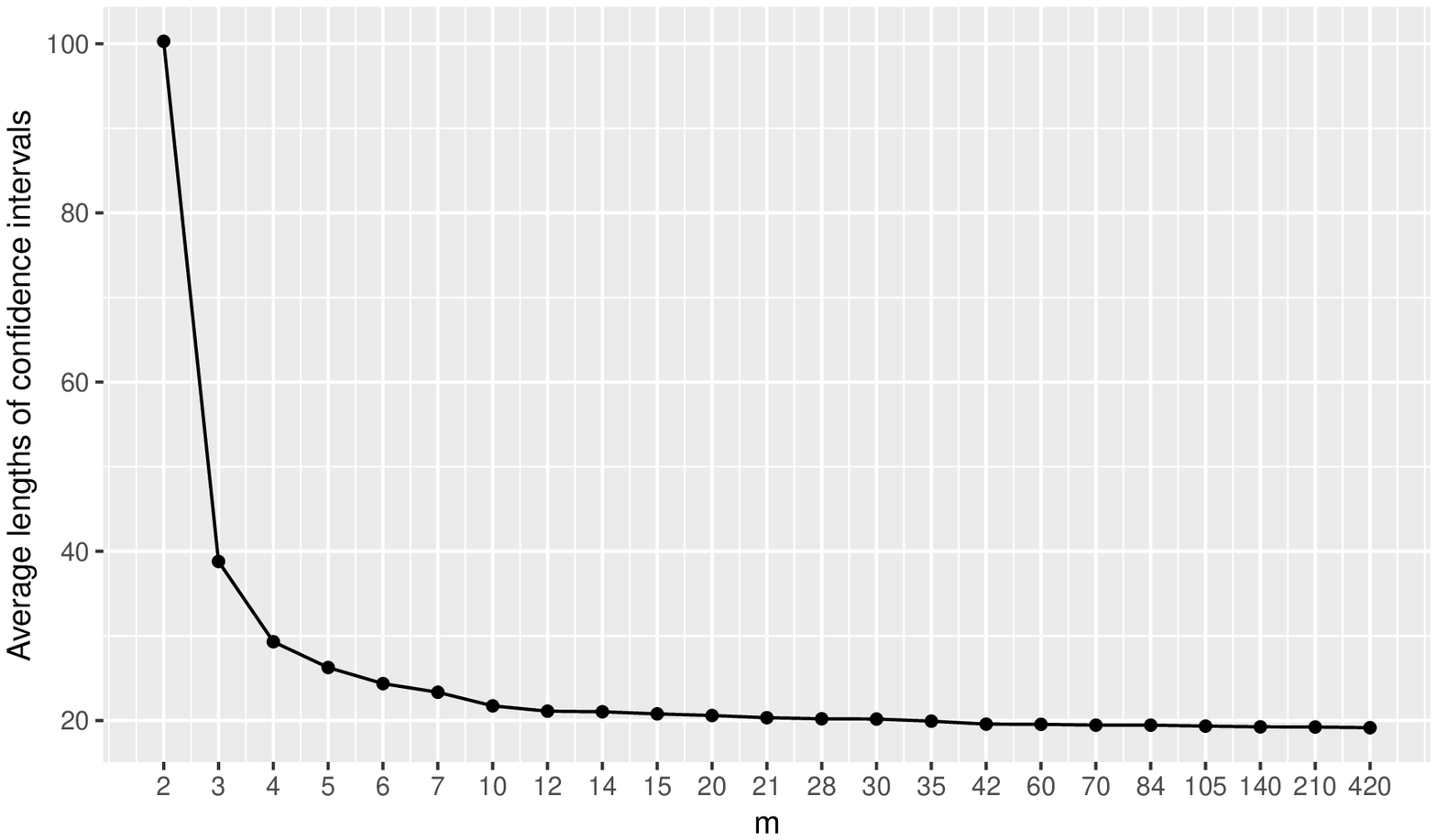}
  \caption{The Expected lengths of Confidence Intervals in \nameref{Sec:Simu2}}\label{fig:Simu2}
\end{minipage}
\end{figure}

\subsection{Simulation 3}\label{Sec:Simu3}

In this simulation, we sample our data under Assumption \ref{as:1} with a special case. We consider an equi-correlated covariance matrix, that is, all off-diagonal elements in $\Sigma_k$ are equal to a constant $\rho \in (0,1]$.

We set $n=5000, k =10,100,500,1000, \rho = 0.1,0.2,\ldots,1.0$ and obtain Table \ref{table:Simu-equi} where AP stands for the values obtained from (\ref{eq:exact}) and SP stands for the values obtained from this simulation.

\begin{table}[htbp]
\centering
\begin{tabular}{|c|c|c|c|c|c|c|c|c|c|c|}
\hline
 $\rho$    & \multicolumn{2}{c|}{1}   & \multicolumn{2}{c|}{0.9} & \multicolumn{2}{c|}{0.8} & \multicolumn{2}{c|}{0.7} & \multicolumn{2}{c|}{0.6} \\ \hline
k    & AP          & SP         & AP          & SP         & AP          & SP         & AP          & SP         & AP          & SP         \\ \hline
10   & 0.46        & 0.48       & 0.48        & 0.45       & 0.51        & 0.48       & 0.53        & 0.53       & 0.56        & 0.56       \\ \hline
100  & 0.16        & 0.14       & 0.16        & 0.19       & 0.17        & 0.17       & 0.18        & 0.16       & 0.20        & 0.25       \\ \hline
500  & 0.07        & 0.09       & 0.07        & 0.06       & 0.08        & 0.09       & 0.08        & 0.08       & 0.09        & 0.08       \\ \hline
1000 & 0.05        & 0.06       & 0.05        & 0.07       & 0.06        & 0.05       & 0.06        & 0.05       & 0.06        & 0.05       \\ \hline
 $\rho$    & \multicolumn{2}{c|}{0.5} & \multicolumn{2}{c|}{0.4} & \multicolumn{2}{c|}{0.3} & \multicolumn{2}{c|}{0.2} & \multicolumn{2}{c|}{0.1} \\ \hline
k    & AP          & SP         & AP          & SP         & AP          & SP         & AP          & SP         & AP          & SP         \\ \hline
10   & 0.60        & 0.61       & 0.64        & 0.64       & 0.69        & 0.72       & 0.76        & 0.76       & 0.84        & 0.84       \\ \hline
100  & 0.21        & 0.21       & 0.24        & 0.25       & 0.28        & 0.26       & 0.33        & 0.28       & 0.45        & 0.48       \\ \hline
500  & 0.10        & 0.08       & 0.11        & 0.10       & 0.13        & 0.12       & 0.15        & 0.17       & 0.22        & 0.23       \\ \hline
1000 & 0.07        & 0.09       & 0.08        & 0.10       & 0.09        & 0.11       & 0.11        & 0.10       & 0.15        & 0.17       \\ \hline
\end{tabular}
\caption{Coverage probability with $n$-case intervals}\label{table:Simu-equi}
\end{table}

In Figure \ref{fig:Simu3-k}, the red lines are SP values and the blue ones are AP values. We can see that the red lines exhibit a smooth decreasing tendency as the $\rho$ increases. The blue lines also have a decreasing tendency and they are close to the red lines, respectively. These plots show that for fixed $k$, the coverage probability decreases as the $\rho$ increases.

\begin{figure}[htbp]
\centering
\begin{minipage}[t]{0.48\textwidth}
\centering
\includegraphics[width=\textwidth]{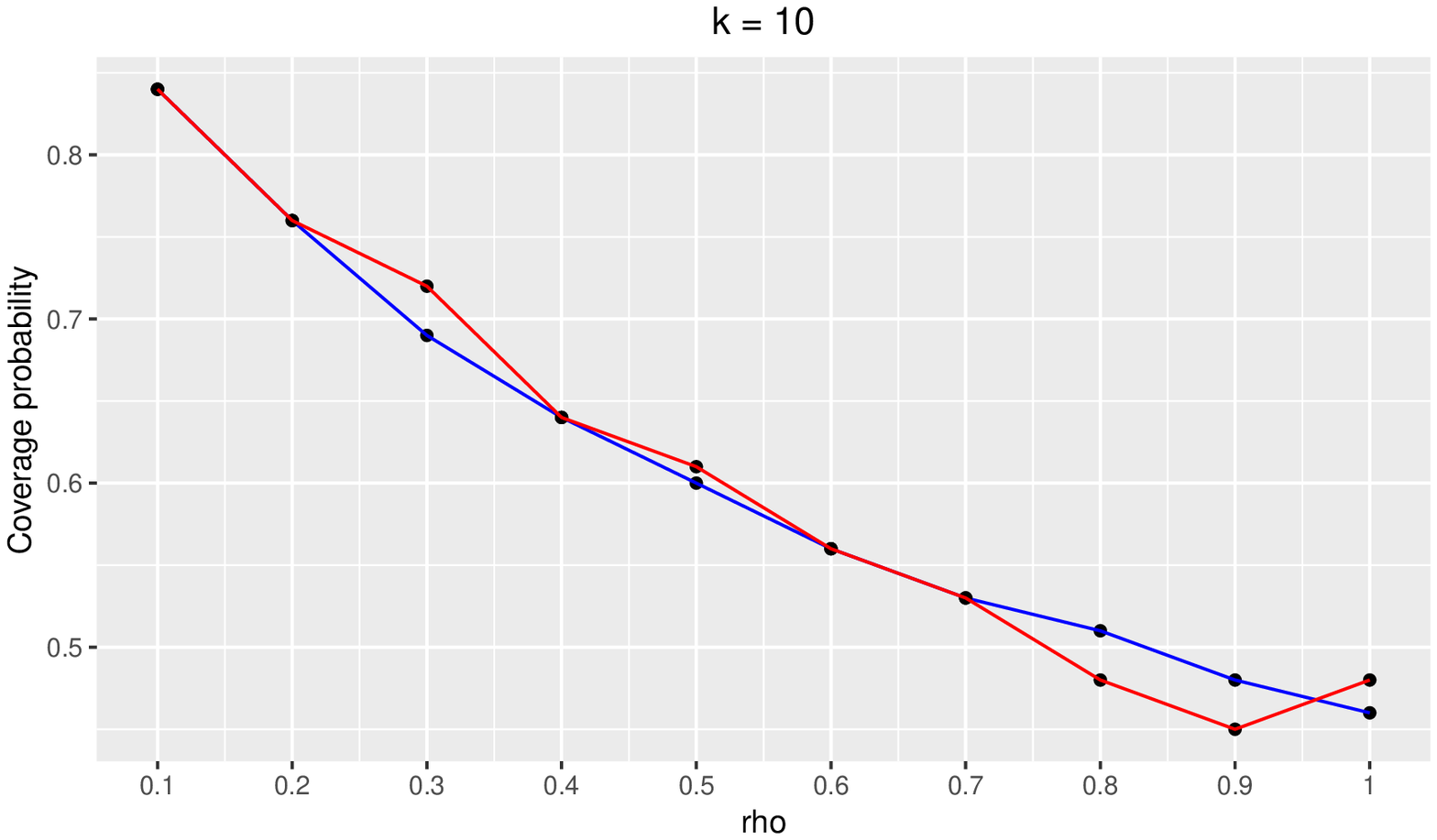}
\end{minipage}
\begin{minipage}[t]{0.48\textwidth}
\centering
  \includegraphics[width=\textwidth]{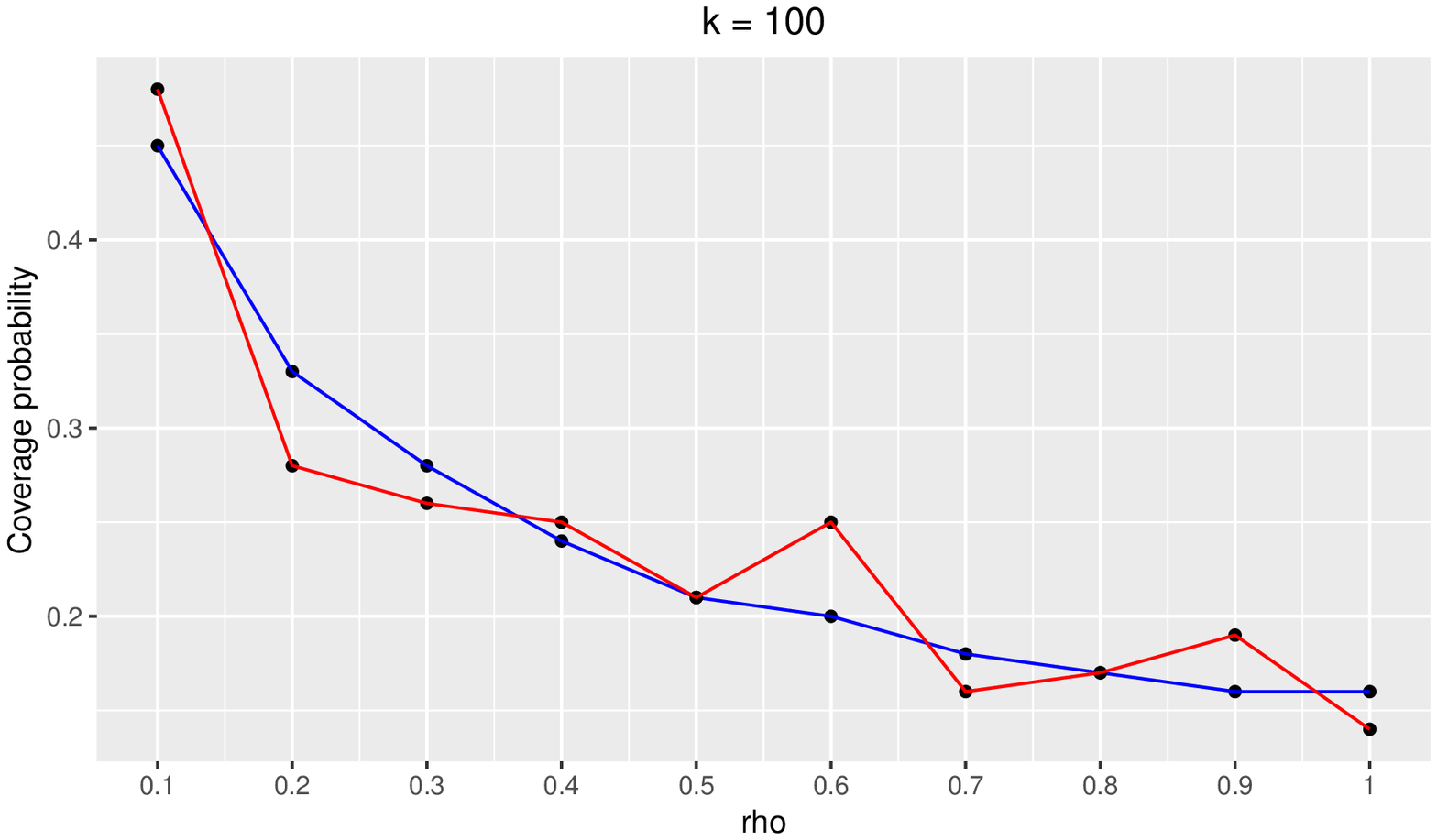}
\end{minipage}
\begin{minipage}[t]{0.48\textwidth}
\centering
\includegraphics[width=\textwidth]{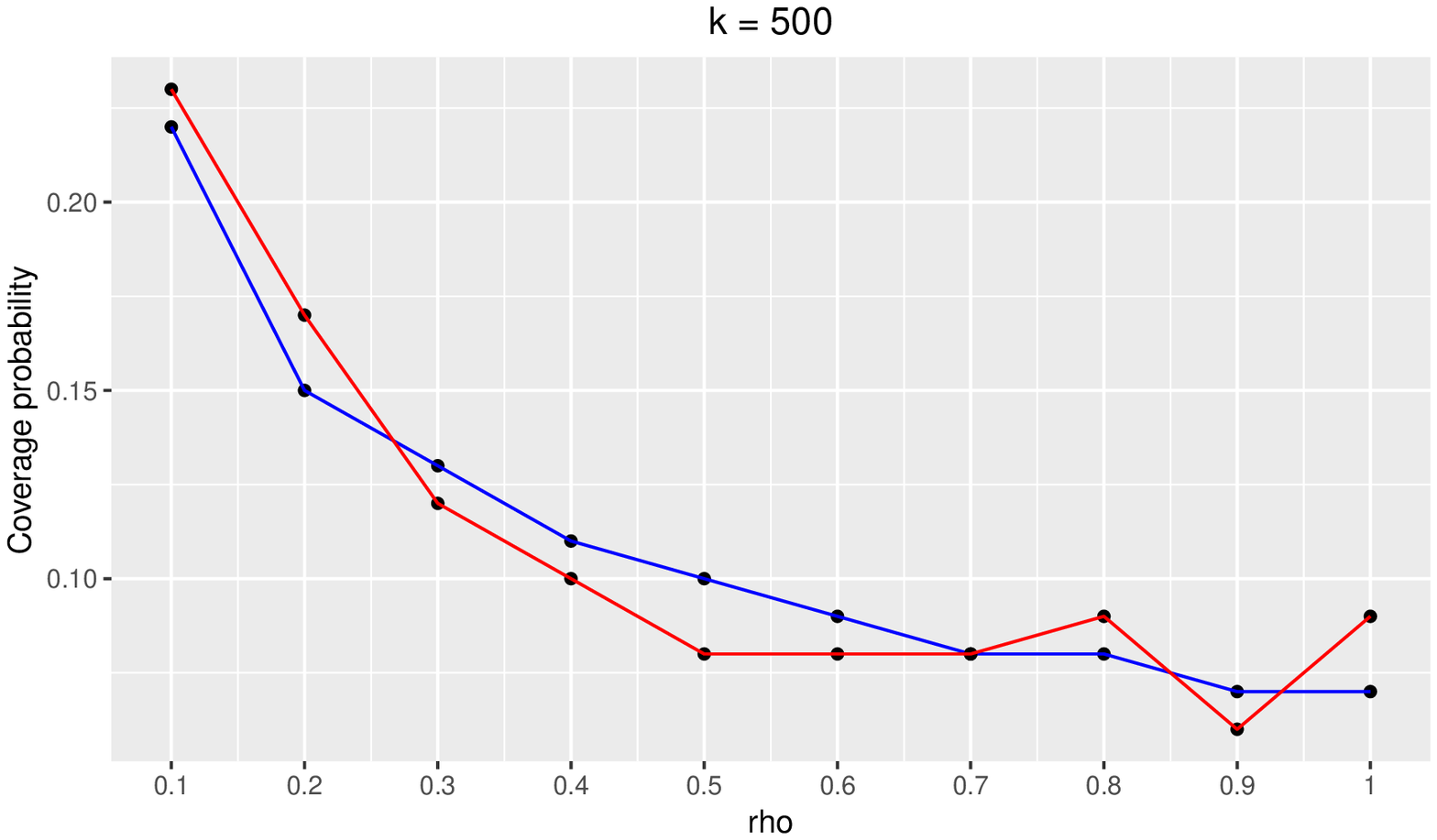}
\end{minipage}
\begin{minipage}[t]{0.48\textwidth}
\centering
  \includegraphics[width=\textwidth]{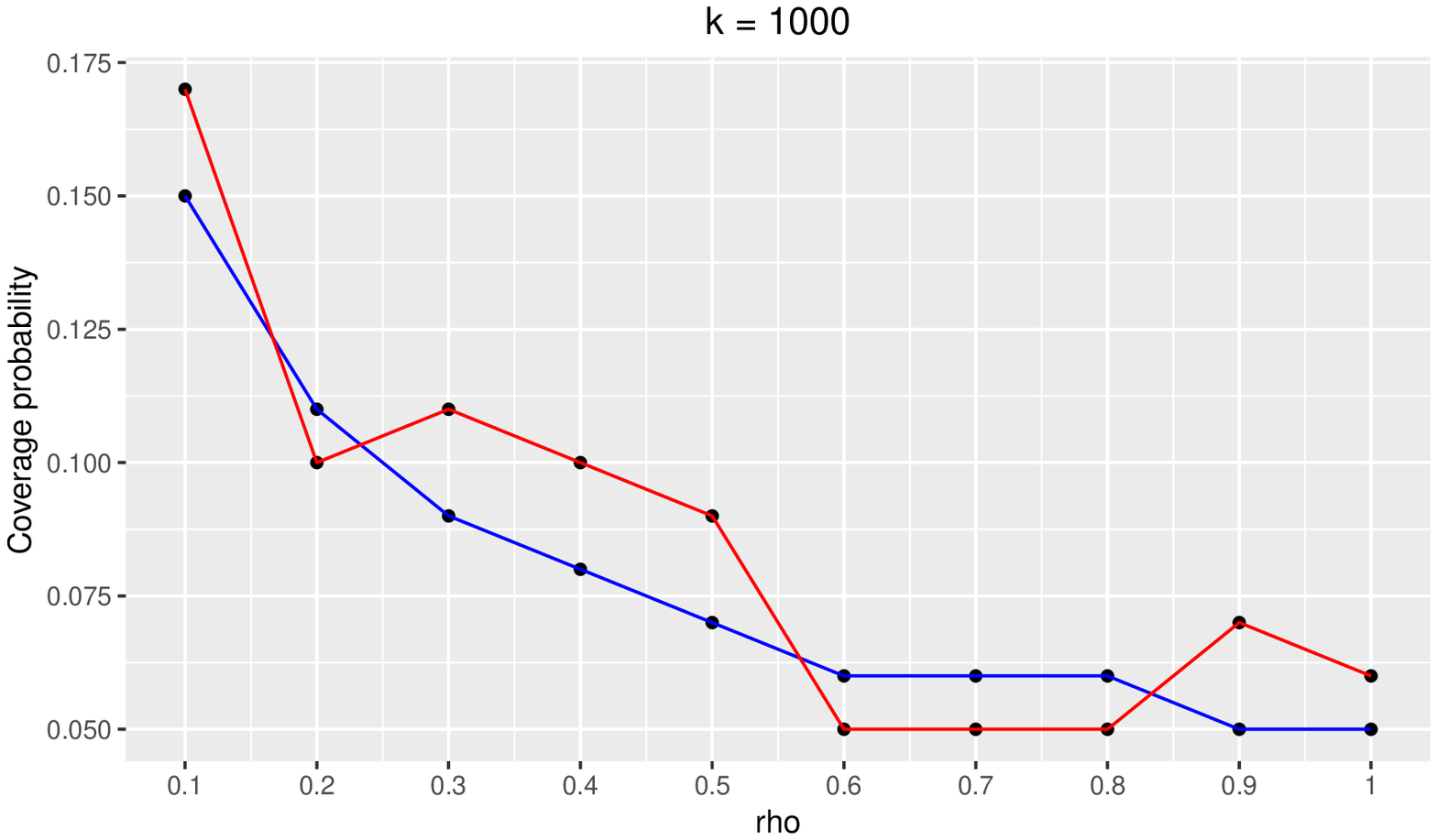}

\end{minipage}
  \caption{Different $k$}\label{fig:Simu3-k}
\end{figure}

In Figure \ref{fig:Simu3-rho}, the number at the top of each plot is the value of $\rho$, the vertical axis is the coverage probability obtained from (\ref{eq:exact}) and the horizontal axis is the values of $k$. We also rescaled the horizontal axes to make them easily understood. From these plots, we know that for fixed $\rho$, the coverage probability decreases as the $k$ increases.

\begin{figure}
  \centering
  \includegraphics[width=\textwidth]{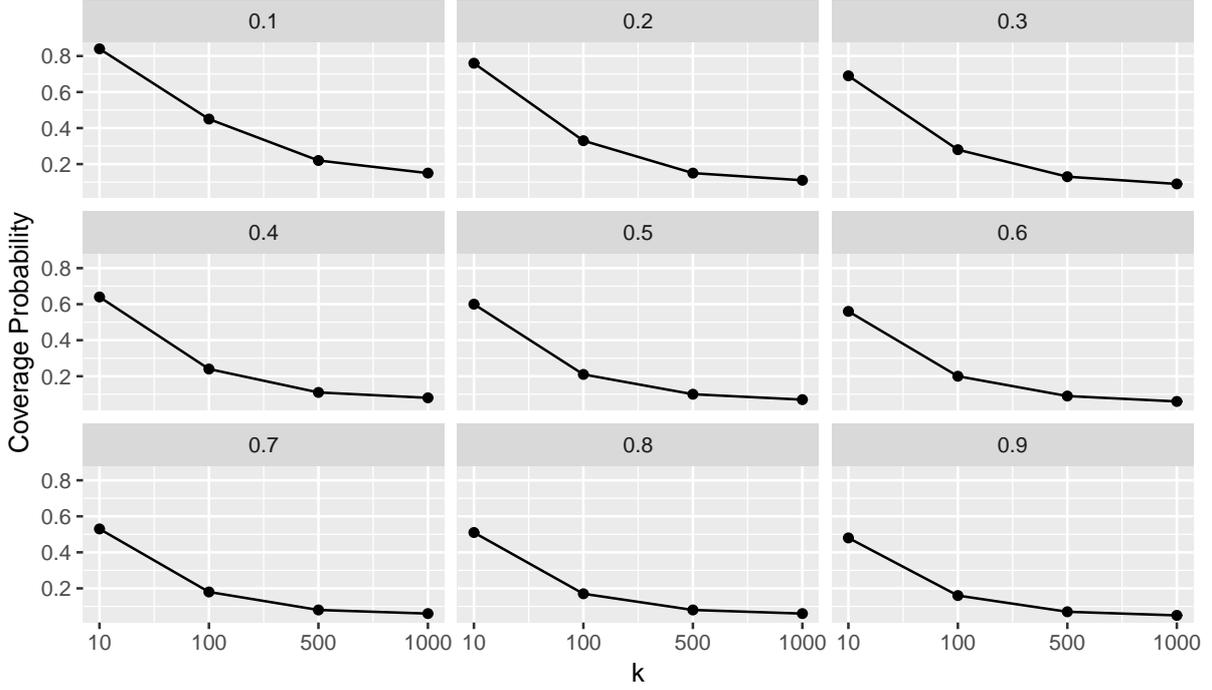}
  \caption{Different $\rho$}\label{fig:Simu3-rho}
\end{figure}

The figures in this simulation all show a decreasing tendency, which is consistent with the information revealed from (\ref{eq:exact}).

\subsection{Simulation 4}\label{Sec:Simu4}

In this simulation, we conduct a real data analysis. This dataset is from the Baltimore site of the Multi-center AIDS Cohort Study (BMACS), which included 400 homosexual men who were infected by the human immunodeficiency virus (HIV) between 1984 and 1991 \cite{Kaslow1987}.

It has 1817 rows and 6 variables. The meaning of each variable is as follows.
\begin{itemize}
  \item  ID. Subject ID
  \item  Time. Subject’s study visit time
  \item Smoke. Cigarette baseline smoking status
  \item age. Age at study enrollment
  \item preCD4. Pre-infection CD4 percentage
  \item CD4. CD4 percentage at the time of visit
\end{itemize}

We extracted the data with Time equals 0.2, then we obtained a subset with 138 observations. We calculate the confidence interval of the mean of the CD4 with this new subset.

In this case, we set $m$ to be $23,46,69,138$ and calculate the related intervals with (\ref{m-case}).

\begin{figure}[htbp]
  \centering
  \includegraphics[width=\textwidth]{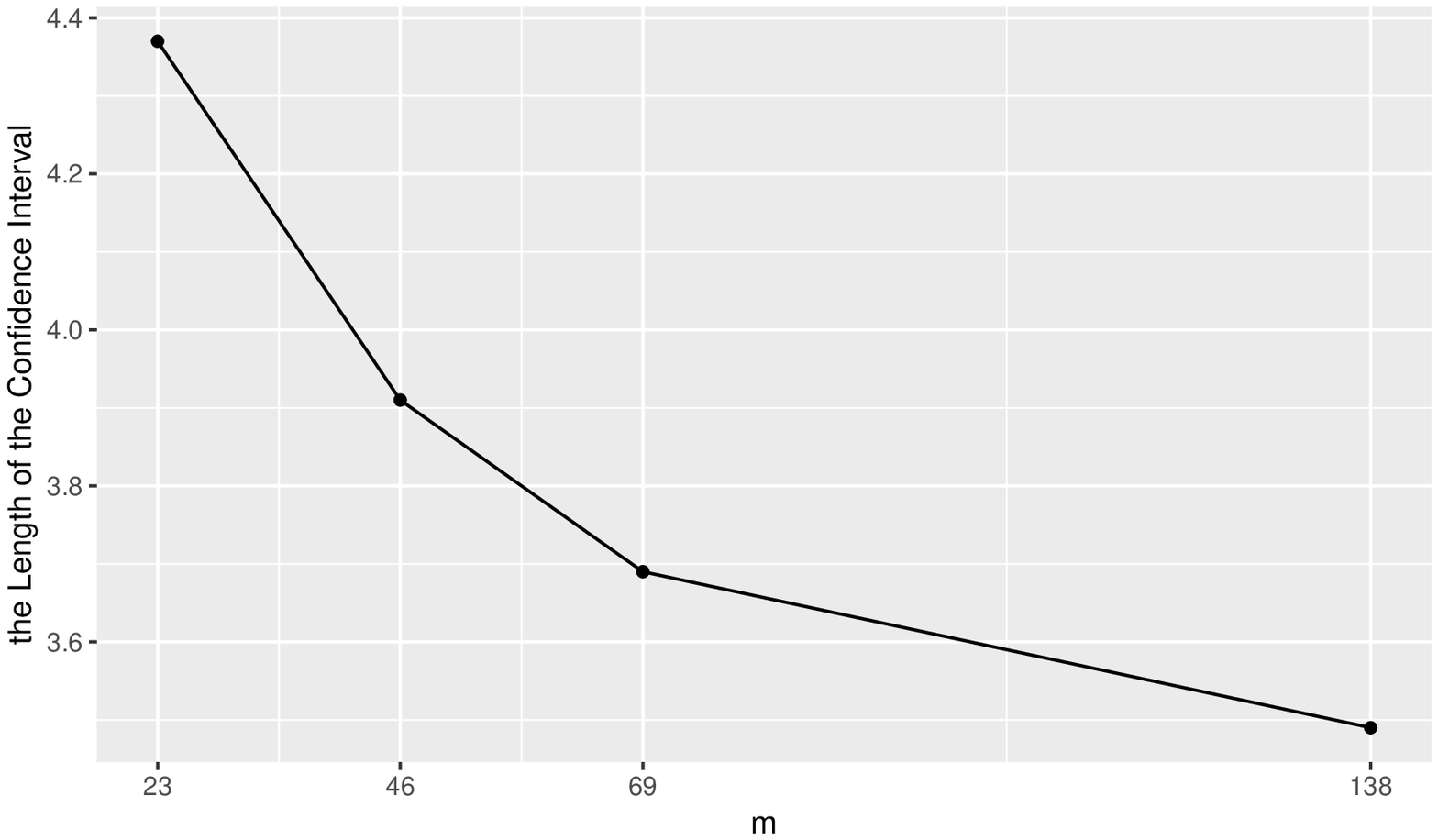}
  \caption{CD4}\label{fig:cd4}
\end{figure}

Figure \ref{fig:cd4} shows four different values of the lengths with above $m$. From Theorem \ref{thm:main}, we learn that generally speaking, a bigger $m$ indicates a shorter interval. In Figure \ref{fig:cd4}, the length decreases gradually as the $m$ increases.

\section{Discussion}\label{sec:discussion}

In this paper we first propose two ways to obtain the confidence interval for the mean parameter $\theta$, i.e., the $m$-case interval and the $n$-case interval, where generally speaking, $n \ge m$. Then we compare these two types of confidence intervals. For $iid$ samples, with fixed confidence level, the expected length of the $n$-case interval is shorter than the one of the $m$-case interval. In other words, we should always use (\ref{n-case}) to calculate the interval under the $iid$ case. Then we propose a special case, under which (\ref{n-case}) is no longer valid. However, (\ref{m-case}) is a feasible method. Although, we can't guarantee that (\ref{m-case}) is an optimal solution.

Interval (\ref{m-case}) is based on equal-sized groups. However, in real practice, our data may be divided into unequal-sized groups. In this case, the task becomes much tougher. The theory we proposed in this paper may not be proper any more when applied under certain circumstances. Future literature may seek new methods to deal with such unequal-sized cases.

In section \ref{sec:extensions}, we provide two different cases where we can utilize our work to get better results.
In addition, \cite{Su2018} is also another application of our results.

\section{Acknowledgements}

I would like to express my appreciation to Prof. Weijie Su, who has instructed me to complete this work. Special thanks to Prof. Xiangzhong Fang, my supervisor in Peking University, who gives me the freedom to do any research I'm interested in. I also appreciate Lynn Selhat, who is an excellent editor and helped me refine this paper. Finally, I acknowledge the support from China Scholarship Council.


\bibliographystyle{unsrt}
\bibliography{paper2-2.bib}
\end{document}